\documentclass[11pt]{article}
\usepackage[margin=1in]{geometry}
\usepackage{fullpage}
\usepackage{tgtermes}
\usepackage[T1]{fontenc}
\usepackage[colorlinks,citecolor=blue,linkcolor=blue,urlcolor=black]{hyperref}
\usepackage{graphicx}
\usepackage{amsfonts,amsmath,amsthm,amssymb,dsfont}
\usepackage{xfrac,nicefrac}
\usepackage{mathdots}
\usepackage{bm,bbm}
\usepackage{url}
\usepackage{paralist}
\usepackage{enumerate}
\usepackage[normalem]{ulem}
\usepackage{xspace}
\xspaceaddexceptions{]\}}
\usepackage{cleveref}
\usepackage{comment}
\usepackage{tabu}
\usepackage{framed}
\usepackage{float,wrapfig}
\usepackage{subfigure}
\usepackage{tikz}
\usepackage[usenames,dvipsnames]{pstricks}
\usepackage[linesnumbered,boxed,ruled,vlined]{algorithm2e}

\theoremstyle{plain}

\newtheorem{theorem}{Theorem}[section]
\newtheorem{lemma}[theorem]{Lemma}
\newtheorem{proposition}[theorem]{Proposition}
\newtheorem{observation}[theorem]{Observation}
\newtheorem{corollary}[theorem]{Corollary}

\theoremstyle{definition}
\newtheorem{definition}[theorem]{Definition}
\newtheorem{remark}[theorem]{Remark}

\newenvironment{reminder}[1]{\bigskip
	\noindent {\bf Reminder of #1.  }\em}{\smallskip}

\def\ShowAuthNotes{1}
\ifnum\ShowAuthNotes=1
\newcommand{\authnote}[2]{\ \\ \textcolor{red}{\parbox{0.9\linewidth}{[{\footnotesize {\bf #1:} { {#2}}}]}}\newline}
\else
\newcommand{\authnote}[2]{}
\fi

\renewcommand{\tilde}{\widetilde}
\newcommand{\R}{\mathbb{R}}
\renewcommand{\Pr}{\operatorname*{\mathbb{P}}}
\newcommand{\E}{\operatorname*{\mathbb{E}}}
\newcommand{\bigo}{O}
\newcommand{\eps}{\varepsilon}
\renewcommand{\epsilon}{\varepsilon}

\newcommand{\poly}{\operatorname{\text{{\rm poly}}}}
\newcommand{\diff}{\operatorname{\text{{\rm diff}}}}

\title{A Massively Parallel Algorithm for  Minimum Weight Vertex Cover}

\author{
Mohsen Ghaffari
\thanks{ghaffari@inf.ethz.ch}\\
ETH Zurich
\and 
Ce Jin
\thanks{jinc16@mails.tsinghua.edu.cn}
\\
Tsinghua University
\and
Daan Nilis
\thanks{nilisdaan@gmail.com}\\
ETH Zurich
}

\begin{document}

	\maketitle
\begin{abstract}
We present a massively parallel algorithm, with near-linear memory per machine, that computes a $(2+\epsilon)$-approximation of minimum-weight vertex cover in $\bigo(\log\log d)$ rounds, where $d$ is the average degree of the input graph. 

Our result fills the key remaining gap in the state-of-the-art MPC algorithms for vertex cover and matching problems; two classic optimization problems, which are duals of each other. Concretely, a recent line of work---by Czumaj et al.\ [STOC'18], Ghaffari et al.\ [PODC'18], Assadi et al.\ [SODA'19], and Gamlath et al.\ [PODC'19]---provides $\bigo(\log\log n)$ time algorithms for $(1+\epsilon)$-approximate maximum weight matching as well as for $(2+\epsilon)$-approximate minimum cardinality vertex cover. However, the latter algorithm does not work for the general \emph{weighted} case of vertex cover, for which the best known algorithm remained at $\bigo(\log n)$ time complexity. 
\end{abstract}

	
\section{Introduction}
Over the past decade, and sparked by the practical successes of popular computing platforms such as MapReduce \cite{Dean2004MapReduceSD}, Hadoop \cite{White2009HadoopTD}, Dryad \cite{Isard2007DryadDD} and Spark \cite{Zaharia2010SparkCC}, the Massively Parallel Computation (MPC) model has emerged as a theoretical abstraction for large-scale parallel computation and it is receiving increasingly more attention from the algorithmic community. In contrast to the fine-grained parallelism found in celebrated models such as the PRAM~\cite{Wyllie1979TheCO}, this new model allows for a higher-level of granularity. In particular, instead of breaking computation into small read or write operations (and basic calculations) that are assumed to happen in lock-step rounds across all processors, the model assumes that each machine can process a small (e.g., polynomially smaller than the entire input) chunk of data per time unit and focuses on the number of rounds of parallelism in such a coarse-grained view. We formally introduce the model in Section~\ref{mpc_definition}.

Shortly after its inception, it has been shown that MPC is at least as powerful as PRAM \cite{Karloff2010AMO, Goodrich2011SortingSA}. Soon after, significantly faster algorithms were developed for a number of important graph problems. The exact speedup depends on the memory that one machine has with respect to the total number of vertices present in the graph. In the setting where each machine has roughly a linear amount of memory with respect to the number of vertices, many of the central problems have been solved in $\bigo(\log \log n)$ rounds. This is true for problems including Maximal Independent Set, nearly Maximum Matching, Maximal Matching, and 2-approximate Minimum Cardinality Vertex Cover \cite{Czumaj2018RoundCF,behnezhad2019exponentially, Ghaffari2018ImprovedMP}. 
The general technique is a certain ``\emph{round compression}'', which starts with an $\bigo(\log n)$ round PRAM or LOCAL \cite{Linial1992LocalityID} model algorithm and successively compresses a considerable fraction of the remaining rounds of the algorithm into just one round of the MPC model. However, these results do not extend to the Minimum Weight Vertex Cover problem and particularly the compression technique fails in managing the deviations created in the weighted case. Indeed, prior to this work, the question of a similar algorithm for the weighted case remained open, and the best known algorithm for the weighted case remained at the classic $\bigo(\log n)$ time that follows from PRAM and LOCAL model literature (e.g. \cite{koufogiannakis2009distributed}). In this paper, we resolve this problem by presenting an $\bigo(\log\log n)$ time algorithm for the weighted case.

In the next sections we formally introduce the MPC model and the current state-of-the-art. Consequently we describe our technical contributions on a high level. 

\subsection{Model Description}
\label{mpc_definition}
In this work we use the MPC (Massively Parallel Computation) model, which can be traced back to descriptions given by Karloff et al.\ \cite{Karloff2010AMO} and Feldman et al.\ \cite{Feldman2008OnDS}, and was later refined by several works \cite{Goodrich2011SortingSA, beame2017communication, andoni2014parallel}. In the most general setting we have the following description: given a problem with input size $N$, there are $M$ machines, each with $S$ words of memory. The size of each machine's memory is assumed to be considerably (e.g., polynomially) smaller than the entire data of the problem. More concretely, the memory size $S$ is assumed to satisfy $S \leq N^{1-\alpha}$ for a constant $\alpha > 0$. Ideally, we want to be able to work with as small as possible memory requirement, per machine. Since the cluster of machines has to be able to store the input, a natural lower bound for the number of machines is $M \geq \frac{N}{S}$; this is typically tight up to a logarithmic factor and we usually assume $M = \widetilde{\Theta}(\frac{N}{S})$. Initially the input is divided arbitrarily among all machines. Computation proceeds in synchronous rounds, where in each round every machine can execute some computation on the data it holds. This local computation is restricted to be of polynomial running time with respect to the local memory size. After this computation, there is a round of communication where each machine can send to every other machine some data -- thus the network graph is the complete graph. The only restriction on the communication is that the total amount of data that one machine sends or receives cannot exceed its memory capacity $S$. The bottleneck in this model is the communication. Therefore the analysis of an algorithm running in the MPC model is focused on the number of rounds.

When considering graph problems, for most problems, the technical difficulty (and the round complexity) of the problem increases as the memory per machine decreases. We will soon discuss instances of this. Considering this effect, in the literature, there is a further distinction of the regimes of the MPC model based on how much memory a machine has relative to the number of vertices, $n$, in the graph:

\begin{itemize}
    \item[(A)] \emph{Strongly super-linear memory regime}: $S \geq n^{1+ \beta}$, for a constant $\beta \in (0,1)$
    \item[(B)] \emph{Near-linear memory regime}: $S \in \widetilde{\Theta}(n)$
    \item[(C)] \emph{Strongly sub-linear memory regime}: $S \leq n^{1-\beta}$, for a constant $\beta\in(0,1)$
\end{itemize}
The goal in the area is to obtain fast algorithms for as small as possible memory requirements, per machine. 

\subsection{State-of-the-Art}
First of all, by positive simulation results from Karloff et al. \cite{Karloff2010AMO} and Goodrich et al. \cite{Goodrich2011SortingSA} it is known that any CREW PRAM algorithm using $\bigo(n^{2-2\epsilon})$ total  memory, $\bigo(n^{2-2\epsilon})$ processors and $t=t(n)$ time can be run in $\bigo(t)$ rounds of MPC. The primary goal in the MPC model is to find algorithms that run strictly (and ideally significantly) faster than their PRAM counterparts. For the problems in focus in this paper---$(2+\eps)$-approximate Vertex Cover, Maximal Matching, and $(1+\eps)$-approximate Maximum Matching---$O(\log n)$ time algorithms are known through this simulation and via classic results in the PRAM and LOCAL model \cite{israeli1986fast, lotker2008improved}. The primary goal for these problems in the MPC setting is to obtain algorithms with time complexity much faster than $O(\log n)$, ideally with as small as possible memory requirements per machine.

For the strongly super-linear memory regime, Lattanzi et al.\ \cite{Lattanzi2011FilteringAM} provided a constant round algorithm to find a maximal matching. Using that algorithm as a building block they also provided constant round algorithms to compute an 8-approximate weighted maximum matching, 2-approximate minimum vertex cover and 3/2-approximate minimum edge cover.

The near-linear memory regime presented much more difficulty and for a number of years, there was no sub-logarithmic time algorithm known for any of these problems. That changed with a breakthrough from Czumaj et al.\ that showed the first sub-logarithmic time algorithm for maximum matching \cite{Czumaj2018RoundCF}.
In that work the authors present an algorithm that computes a $(2+\epsilon)$-approximate maximum matching in $\bigo((\log \log n)^2)$ MPC rounds. Later this result has been improved and simplified by Assadi et al.\ \cite{Assadi2017CoresetsME} and Ghaffari et al.\ \cite{Ghaffari2018ImprovedMP} to reach an $\bigo(\log \log n)$ round algorithm that computes a $(1+\epsilon)$-approximate maximum matching.
Moreover, the method of Ghaffari et al.\ \cite{Ghaffari2018ImprovedMP} also provides a $(2+\epsilon)$-approximate minimum cardinality vertex cover. Gamlath et al.\ \cite{gamlath2019weighted} showed an extension of the matching problem to the weighted case, providing a $(1+\epsilon)$-approximate maximum weight matching in $\bigo(\log \log n)$ rounds. However, the weighted case of the vertex cover problem has remained open and the best known algorithm remains at the $O(\log n)$ complexity that follows from the PRAM literature.

\medskip
\noindent \textbf{A remark regarding the strongly sub-linear memory regime.}
The case of the strongly sub-linear memory regime appears to be considerably harder and there is no known $\poly(\log\log n)$ round algorithm for the general case of any of the above problems. The best known result is a work of Ghaffari and Uitto \cite{Ghaffari2018SparsifyingDA} that provides $\tilde \bigo(\sqrt{\log n})$-round algorithms for maximal independent set, maximal matching, 2-approximation of minimum vertex cover, and $(1+\epsilon)$-approximation of maximum matching. For special graph families, concretely trees and low arboricity graphs, $\poly(\log\log n)$ round algorithms were provided by Behnezhad et al.~\cite{behnezhad2019massively}.

\subsection{Our Contributions}
Our main result, which positively answers the open question for the weighted case of vertex cover in the near-linear memory regime, is as follows:

\begin{theorem}
There is a randomized MPC algorithm, with $\tilde{O}(n)$ memory per machine, that computes a $(2+\epsilon)$-approximate minimum weight vertex cover in $\bigo(\log\log d)$ rounds in any input graph with $n$ vertices and average degree $d$, with high probability\footnote{As standard, we use the phrase \emph{with high probability}, or the abbreviation w.h.p., to indicate that an event happens with probability at least $1-\frac{1}{n^c}$ for any desirable constant $c>0$.}. 
\end{theorem}

We emphasize that our round complexity is a function of the \emph{average} degree $d$, instead of the maximum degree $\Delta$.
We are not aware of any prior work in MPC where the round complexity is a function of the average degree.

Let us briefly discuss the method: Our algorithm follows a similar outline with the algorithm by Ghaffari~et~al.~\cite{Ghaffari2018ImprovedMP} for the unweighted case, and more generally the \emph{round compression} idea introduced by Czumaj et al. \cite{Czumaj2018RoundCF}, but with some crucial and important changes that allow us to handle the weighted case.

The general round compression technique works roughly as follows: randomly partition the vertices among a small set of machines and simulate many iterations of a suitable LOCAL algorithm on the induced subgraphs in these machines. Then communicate the results and repeat for a few steps. The power lies in the possibility to simulate up to a constant fraction of the iterations of the LOCAL algorithm, by just working on the randomly partitioned graph and without any further communication among the machines, since the error incurred from neglecting cross-partition interactions is very small (which can be shown using concentration inequalities). 

Compared to the unweighted case, it is more difficult to analyze the behaviour of the progress that is being made during the algorithm, simply because there is an additional factor that influences the behaviour of a vertex, namely its weight. The behaviour of a vertex in the unweighted setting is completely dependent on its degree, which is moreover possible to estimate after taking a random subgraph. This property does not hold when vertex weights are added---due to the inherent deviations in the related random variables---and hence additional ideas are needed to make sure that we can still simulate the vertices well and moreover, that our algorithm makes enough progress in each step. 

After introducing the necessary background, in Section~\ref{sec3.2} we will give an overview of our novel ideas for resolving the above-mentioned issues in the weighted case.
We suspect that some of the ideas presented there for handling weights in round compression might find applications in other instances of round compression, where we have to deal with weights (and where natural sampling ideas face deviation issues).

\medskip
\noindent\textbf{Implications for Congested Clique.}
 A very closely related model, which has received significant attention over the past few years from the distributed computing community, is the \emph{congested clique} model~\cite{Lotker2005MinimumWeightST}. This model was initially proposed to capture computing on overlay networks. The network is presented as a fully connected graph, where on each node there is a machine and the machines can communicate in an all-to-all fashion. The precise method of communication is abstracted away -- there could be a direct link between all machines or a routing protocol to make communication possible. Therefore, the communication is considered to be the bottleneck and only small messages can be sent (of size $\bigo(\log n)$, with $n$ the number of nodes in the graph). Computation proceeds again in synchronized rounds where each machine executes some local computation and then sends messages to other machines. The local memory and computation power are assumed to be unlimited. 
 
 In \cite[Theorem 3.2]{Behnezhad2018BriefAS} the authors show a two-way simulation to show that the near-linear memory MPC setting, in their words  \emph{semi-MapReduce}, is equivalent to congested clique.  Thus, our result also implies an $O(\log\log d)$ round algorithm for $(2+\eps)$ approximate minimum weight vertex cover in the congested clique model.

\medskip
\noindent \textbf{Roadmap.}
The remaining sections are structured as follows: Section~\ref{sec:prelim} presents some basic preliminaries and notation. In Section \ref{chap:algo}, we provide an in depth description of our algorithm. Then, Section \ref{chap:analysis} provides the detailed analysis of the memory requirements, round complexity, and accuracy achieved by our algorithm.
\section{Preliminaries}
\label{sec:prelim}
We denote an undirected graph $G = (V,E)$ by its vertex set $V$ and its edge set $E$, where an edge $e\in E$ is an unordered pair of vertices, e.g. $e=(u,v)$ represents an undirected edge between the vertices $u, v \in V$.
Furthermore, we denote the degree of a vertex with $d(v)$ and the maximum degree of any vertex present in the graph is $\Delta$.
In this paper we consider graphs  with vertex weights, denoted $w(u)\in \R^+$ for $u\in V$.
For a vertex subset $V'\subseteq V$, we will use the notation $E[V']$ to denote $\{e=(u,v)\in E \mid u\in V', v\in V'\}$, the edges in the induced subgraph of $V'$.
We use $E[V';V'']$ to denote $\{e=(u,v)\in E \mid u\in V', v\in V''\}$.


For asymptotic notations $\bigo, \Theta, \Omega$, an additional tilde hides polylogarithmic factors. For example, $\widetilde{\bigo}(f)$ denotes $\bigo(f\cdot \poly \log (f))$.

We will frequently use the following form of Chernoff bounds to bound the tails of a sum of independent random variables.
\begin{theorem}[Chernoff bounds]
Let $X = \sum_{i=1}^n X_i$ be the summation of independent random variables, each assuming values in $[0,1]$. 
Let $\mu = \E(X)$. Then
\begin{itemize}
    \item $\Pr(|X-\mu| \geq \delta \mu) \leq 2\exp{(-\delta^2\mu/3 )}$ for $0 \leq \delta \leq 1$.
    \item $\Pr(|X-\mu| \geq \delta \mu) \leq 2\exp{(-\delta\mu/3 )}$ for $\delta > 1$.
\end{itemize}
\end{theorem}
More information about the background and applications of these Chernoff bounds can be found in \cite{Dubhashi2009ConcentrationOM}.

\section{The Algorithm}
\label{chap:algo}
As discussed earlier in the introduction, our MPC algorithm follows the framework of \cite{Ghaffari2018ImprovedMP} and uses the powerful \emph{round compression} technique first introduced in \cite{Czumaj2018RoundCF}. Recall that a critical part of such an MPC algorithm is a centralized/LOCAL algorithm that allows for efficient simulation under random sampling.

In the following section we describe the centralized algorithm we will use and analyse its approximation guarantee. 
Then in the next section we give an overview of our MPC algorithm and highlight the differences with the algorithm for the unweighted case \cite{Ghaffari2018ImprovedMP}. In the last section we outline the full MPC algorithm.

\subsection{Centralized Algorithm}
\label{sec3.1}
We first describe a centralized algorithm for computing a $(2+\epsilon)$-approximate weighted vertex cover, using the standard primal-dual framework; this approach can be traced back to the first studies on approximating the vertex cover problem \cite{hochbaum1982approximation, bar1981linear}. 


\begin{figure}
\begin{center}
\begin{tabular}{c c c c c c}
\label{int program}
& Primal & \\
$\min$ & $\sum_{v\in V} z_v \cdot w(v)$ & \\
s.t. & $z_u + z_v \geq 1$ &  $\forall (u,v) \in E$ \\
& $z_v \geq 0$ &  $\forall v \in V$ \\
\\
& Dual & \\
$\max$ & $\sum_{e \in E} x_e $ &\\
s.t. & $\sum_{e\ni v} x_e \leq w(v)$ & $\forall v \in V$ \\
&$ x_e \geq 0$ & $\forall e \in E$\\
\end{tabular}
\end{center}
\caption{Linear programming relaxation for MWVC}
\label{fig:primal dual}
\end{figure}


We maintain the dual variables $\{x_e\}_{e\in E}$ which form a valid \emph{fractional matching}, i.e., they satisfy the dual constraints $\sum_{e \ni v} x_e  \le w(v)$ for all $v\in V$. 
Every vertex has a status of being \emph{active} or \emph{frozen}, indicating whether this vertex is still participating in the algorithm.
We say an edge is \emph{active}, if and only if both of its endpoints are active.
We start with a valid fractional matching, and set all vertices to be \emph{active}.
Then, we slowly increase the dual variable $x_e$ of every active edge $e$, while not violating the dual constraints. When the dual constraint of a vertex becomes near-tight, we freeze this vertex and include it in our vertex cover solution.
In the end, we can show by weak LP-duality that this solution is indeed a $(2+O(\epsilon))$-approximation.

Algorithm \ref{algo:local det mwvc} implements such a primal-dual scheme. 
The choices of the initial weights $x_e$ and the thresholds $\mathcal{T}_{v,t}$ are not specified in the description. We will specify them later when we simulate this centralized algorithm in the MPC model and compare their behaviour. Next, we analyze the approximation guarantee of this centralized algorithm.

\begin{algorithm}
\SetAlgoLined
\begin{enumerate}
    \item Input: graph $G=(V,E)$, weight function $w: V\to\R^+$
    \item Initialization:
         let $\{x_{e,0}\}$ ($x_{e,0}>0$ for all $e\in E$) be an arbitrary valid fractional matching
    \item Let $\mathcal{T}_{v,t}$ be arbitrary numbers from interval $[1-4\eps,1-2\eps]$, for all $v\in V$ and integers $t\ge 0$
    \item While at least one edge is active, iterate $t\gets 0,1,\dots$:
    \begin{enumerate}
        \item For each active vertex $v$ satisfying $y_{v,t} := \sum_{e \ni v} x_{e,t} \geq \mathcal{T}_{v,t} \cdot w(v)$: freeze $v$ and its incident edges \label{algoline:freeze central}
        \item For each active edge $e$: $x_{e,t+1} := x_{e,t}/(1-\epsilon)$
        \item For each frozen edge $e$: $x_{e,t+1} := x_{e,t}$
    \end{enumerate}
    \item Return all frozen vertices as a vertex cover
\end{enumerate}
\caption{A generic centralized MWVC algorithm}
\label{algo:local det mwvc}
\end{algorithm}

\begin{observation}
\label{obs:central valid}
For all $t\ge 0$, the dual constraint $\sum_{e\ni v}x_{e,t} \le w(v)$ is satisfied for all $v\in V$. In other words, Algorithm~\ref{algo:local det mwvc} maintains a valid fractional matching. 
\end{observation}
\begin{proof} We do a proof by induction on $t$.
The validity for $t=0$ is ensured by the initialization requirement. For the inductive step, assume that $\{x_{e,t}\}_{e\in E}$ is a valid fractional matching. By the start of iteration $t+1$,
for any active vertex $v$,
\begin{equation*}
\sum_{e \ni v}x_{e,t+1} \leq \sum_{e \ni v}\frac{x_{e,t}}{ 1-\epsilon} < \frac{\mathcal{T}_{v,t} w(v)}{1-\eps} \le \frac{(1-2\epsilon)w(v)}{(1-\epsilon)} < w(v).
\end{equation*}
where the second inequality follows from the fact that vertex $v$ was not frozen in iteration $t$.
 For any frozen vertex $v$,
 \[
\sum_{e \ni v}x_{e,t+1}  = \sum_{e \ni v}x_{e,t} \le w(v).
\qedhere
\]
\end{proof}

\begin{lemma}[Weak LP-duality]
Let $OPT$ be the total weight of the minimum weight vertex cover $C^*$ of graph $G=(V,E)$, and $\{x_e\}_{e\in E}$ be any fractional matching of $G$. Then $OPT \ge \sum_{e\in E}x_e$.
\end{lemma}
\begin{proof}
Observe that \[OPT = \sum_{v\in C^*} w(v) \ge \sum_{v\in C^*} \sum_{u:(u,v)\in E} x_{(u,v)} \ge \sum_{e\in E} x_{e}.\qedhere\]
\end{proof}
\begin{proposition}
\label{prop:central apx}
When Algorithm \ref{algo:local det mwvc} terminates, it returns a vertex cover $C$ which satisfies 
$$ w(C) \leq (2+10\epsilon) OPT,$$
where $OPT$ is the weight of a minimum weight vertex cover.
\end{proposition}

\begin{proof}
We first claim that the returned set of frozen vertices forms a valid vertex cover. This follows from the fact that the algorithm only terminates when all edges have been frozen, i.e. when they contain at least one vertex that is frozen. Therefore the set of frozen vertices covers all the edges.

Next we will relate the weight of the fractional matching to the size of the vertex cover and use LP-duality to prove the claimed approximation ratio. Denote the value of the final fractional matching by $W_M = \sum_{e\in E} x_e$. For every vertex $v$ in the returned vertex cover $ C$, 
$$
 y_v = \sum_{e \ni v} x_e \geq \mathcal{T}_{v,t^*}w(v) \ge (1-4\epsilon)w(v),
$$
where $t^*$ is the iteration in which $v$ became frozen.
As each edge can be covered at most twice, once per endpoint, we have
$$
2 W_M = 2\sum_{e\in E} x_e \geq \sum_{v \in C} \sum_{e \ni v} x_e \geq (1-4\epsilon) w(C),
$$
where $w(C)$ denotes $\sum_{v\in C}w(v)$. 
Then, by weak LP-duality,
\begin{equation*}
w(C) \leq \frac{2}{1-4\epsilon}W_M \leq \frac{2}{1-4\epsilon} OPT<(2+10\eps)OPT. \qedhere
\end{equation*}
\end{proof}

%
%

\subsection{Overview}
\label{sec3.2}
Before presenting our MPC algorithm, let us briefly review the algorithm for the unweighted case from Ghaffari et al. \cite{Ghaffari2018ImprovedMP}. 

\medskip
\noindent\textbf{A recap on the approach of Ghaffari et al. \cite{Ghaffari2018ImprovedMP}.}  Their algorithm proceeds in \emph{phases}, where the machines only communicate after each phase.
At the start of a phase, the vertices are partitioned uniformly at random among $m=\sqrt{\delta}$ machines, where $\delta$ is an upper bound on the current maximum degree (in the induced subgraph of nonfrozen vertices). Next, each machine gathers the induced subgraph on the vertices it received and simulates the LOCAL primal-dual algorithm (Algorithm~\ref{algo:local det mwvc}, with $w(v) = 1, \forall v \in V$) on this subgraph, where 
the initialization of the fractional matching is taken as $x_{e,0} = 1/n$.
This LOCAL algorithm is simulated by only inspecting the \emph{local} neighborhood of each vertex, i.e. the neighbors that landed on \emph{the same machine}.
When checking the dual constraints (Line (\ref{algoline:freeze central}) of Algorithm \ref{algo:local det mwvc}), the algorithm uses the (scaled) total weight of incident edges from local neighbors, which is an unbiased estimate of the total incident weight on the full graph.
These estimates are sharply concentrated, and hence through the use of \emph{random thresholds} $\mathcal{T}_{v,t}$, with good probability, the behaviour of this simulation is very close to the behaviour of the LOCAL algorithm on the full graph for all iterations.\footnote{We refer readers to {\cite[Section~4.2]{Ghaffari2018ImprovedMP}} for more intuition on their random thresholding technique and a discussion of its necessity.} 
Their algorithm proceeds by reducing the maximum degree until $\delta < \poly\log n$, at which point the algorithm terminates in one more step solving the remaining instance (with only $\tilde O(n)$ edges) on one machine.\\

Our algorithm for the weighted case uses the framework of Ghaffari~et~al.~\cite{Ghaffari2018ImprovedMP}, but has a few key differences. In the following, we give a high-level description of our new techniques.

\medskip \noindent\textbf{Non-uniform initialization of edge weights.}  First of all, we use a different initialization of the LOCAL algorithm. Instead of using the standard initialization $x_{(u,v)}=1/n$, we use $x_{(u,v)} := \min\left \{\frac{w(v)}{d(v)}, \frac{w(u)}{d(u)}\right \}$.\footnote{The actual initialization used in our MPC simulation is slightly different regarding the definition of residual degrees $d(u)$, for technical reasons (see Remark~\ref{remark:deg}).}
This non-standard initialization will be crucial to the analysis of our MPC algorithm.
We first give a succinct analysis of this initialization in the centralized setting. 

\begin{proposition}
The initialization $x_{(u,v),0} := \min\left \{\frac{w(v)}{d(v)}, \frac{w(u)}{d(u)}\right \}$ is valid. Moreover, Algorithm~\ref{algo:local det mwvc} terminates after $\bigo(\log \Delta)$ iterations under this initialization.
\label{prop:central initialization}
\end{proposition}
\begin{proof}
For every vertex $v$, $\sum_{e \ni v}x_{e,0} \leq d(v) \cdot \frac{w(v)}{d(v)} = w(v).$

For the running time, consider any edge $e=(u,v)$ and w.l.o.g. assume its initial weight is $x_{e,0} = \frac{w(u)}{d(u)}$. If $e$ is active after $\log_{(1/(1-\epsilon))}(\Delta)$ iterations, we have that $x_e \geq  w(u)$, which violates the dual constraint and cannot happen. Hence, the algorithm terminates after 
$\log_{(1/(1-\epsilon))}(\Delta) \in \bigo(\log \Delta)
$ iterations.
\end{proof}

One can see that the standard initial assignment $x_{e,0} = 1/n$ also yields a correct LOCAL algorithm for the weighted vertex cover problem, assuming the weights of vertices are rescaled so that $w(v)\ge 1$.
However, the running time of this LOCAL algorithm would depend on the size of the maximum vertex weight: $\bigo(\log(Wn))$ with $W =\max_{v\in V}w(v)$, which is undesirable.

One might want to use $x_{(u,v)}:=\min\{\frac{w(v)}{\Delta},\frac{w(u)}{\Delta}\}$ instead of $x_{(u,v)}:=\min\{\frac{w(v)}{d(v)},\frac{w(u)}{d(u)}\}$ for initialization. The former has smaller weights, and hence causes the primal-dual algorithm to make progress slower, though this difference is not obvious in the LOCAL model---the former initialization achieves the same time bound as stated in Proposition~\ref{prop:central initialization}. However, when performing MPC simulation of the LOCAL algorithm, we could only achieve $O(\log \log \Delta)$ round complexity when using the former way of initialization.
Using the latter one, we can achieve round complexity $O(\log \log d)$ (where $d$ is the average degree) as claimed in the main result.

\medskip \noindent\textbf{Analysis of progress via orienting edges.} 
An integral part of applying round compression is that the graph gets sparsified, and hence we need to quantify this progress.
In Ghaffari~et~al.'s unweighted algorithm~\cite{Ghaffari2018ImprovedMP},  each nonfrozen edge has the same weight $x_t = (1/n)/(1-\eps)^t$, so the number of nonfrozen neighbors of any vertex is upper bounded by $1/x_t$, which gives a natural characterization of the sparsity of the remaining graph.
For the weighted case, that characterization of progress does not hold anymore, due to our non-uniform initialization. 

Instead, we will use an \emph{orientation} argument. We orient every edge $(u,v)$ from $u$ to $v$ if $\frac{w(u)}{d(u)}<\frac{w(v)}{d(v)}$. Then, since every edge $e$ outward from vertex $u$ has initial weight equal to $\frac{w(u)}{d(u)}$, we can give a natural upper bound on the \emph{out-degree} of $u$,  and analyze the out-degree shrinking over time.
Although we do not have control on the (undirected) degrees of vertices, our upper bounds on the out-degrees still allow us to bound the total number of remaining edges and measure the progress of the algorithm.

\medskip \noindent\textbf{Not simulating low degree vertices.} 
Recall that the LOCAL algorithm needs to be simulated after taking a random sample of the vertices, and the accuracy of the simulation relies on necessary concentration bounds of the incident edge weight for each vertex. 
In the weighted case, low degree vertices can cause difficulties for proving such concentration, since (1) they can have big initial weights (due to our non-standard initialization), which introduce large deviation; and, (2) the sampling rate is not enough for their neighborhood. (For the algorithm of the unweighted case this is not an issue, since a low degree vertex cannot become frozen during the early stages of the LOCAL algorithm) 


To alleviate the issue we divide the vertices in two classes at the start of a phase: in $V^{high}$ are all the nonfrozen vertices with a \emph{high degree}, defined as vertices $v$ such that $d(v) \geq d^{0.95}$, where $d(v)$ is the degree of vertex $v$ with respect to nonfrozen neighbors only,  and $d:=\frac{1}{n}\sum_{v\in V\text{ nonfrozen}} d(v)$.\footnote{Note that $d$ is not quite the ``average'' degree; the denominator is always $n$ regardless of the number of nonfrozen vertices $v$.} 
The other nonfrozen vertices of low degree are called \emph{inactive}, and are gathered in $V^{inactive}$.
Then only the vertices in $V^{high}$ are partitioned and simulated in this phase.
In the analysis we will see that despite only simulating a subset of the vertices the algorithm still makes enough progress in one phase on reducing the average degree to reach an overall running time of $\bigo(\log\log d)$ MPC rounds.


\medskip \noindent\textbf{Other changes in our analysis.} 
Since we need to deal with weights and degrees in  our sampling/simulation arguments, the required concentration bounds are more delicate than the previous work~\cite{Ghaffari2018ImprovedMP}. 

To simplify some parts of the analysis, we make another modification of the algorithm. Recall that in Ghaffari~et~al.'s \cite{Ghaffari2018ImprovedMP} algorithm, the (scaled) total incident weight of local neighbors was used as an estimate of the actual total incident weight on the full graph. 
This estimate is unbiased, and could have error on either of the two directions. One of them is easy to deal with, while the other one requires a much more difficult analysis (see the discussion in \cite[Section 4.4.4]{Ghaffari2018ImprovedMP} on ``late-bad vertices''). In our algorithm, we 
simply introduce a bias term to the estimator (Line~\ref{algoline:inner freeze} of Algorithm~\ref{algo:mpc mwvc}), so that with high probability it only has one-sided error comparing to the actual total incident weight.
This will make the analysis easier when we compare the behaviour of our MPC simulation with the centralized algorithm in Section~\ref{section:approx ratio}.
 \begin{algorithm}
 \SetAlgoLined
 \begin{enumerate}
     \item Input: graph $G=(V,E)$, weight function $w: V\to\R^+$
     \item While $d:= \frac{1}{n}\sum_{v\in V\text{ nonfrozen}} d(v) > \log^{30} n$: \label{algoline:phase}
     \begin{enumerate}
         \item Let $V^{high} \gets \{v \text{ nonfrozen}\mid d(v) \geq d^{0.95}\}$, $V^{inactive} \gets \{v \text{ nonfrozen}\mid d(v) < d^{0.95}\}$ 
         \label{algoline:alive vert}
         \item Compute residual weights for all $v\in V^{high}$: $w'(v) \gets w(v) - \sum_{e \ni v, \text{ frozen}} x_e^{MPC}$
         \label{algoline:residual weight}
         \item Initial edge weights for all $e=(u,v) \in E[V^{high}]$: $x_{e,0}^{MPC} :=  \min \left \{ \frac{w'(u)}{d(u)}, \frac{w'(v)}{d(v)}\right \}$
         \item Let $\mathcal{T}_{v,t}$ be independent random numbers uniformly chosen from $[1-4\eps,1-2\eps]$, for all $v\in V^{high}$ and $0\le t< I$
         \item Set number of machines $m:=\sqrt{d}$, and number of iterations $I:=\frac{\log m}{10\log 15}$ \label{algoline:nb machines}
         \item Partition $V^{high}$ into $m$ sets $V_1, \dots ,V_m$ by assigning each vertex to a machine independently and uniformly at random \label{algoline:def vi}
         \item For each $i \in \{1, \dots,m\}$ in parallel, iterate $t\gets 0,1,\dots, I-1$:
         \begin{enumerate}
             \item For each active $v \in V_i$ satisfying  $\tilde{y}^{MPC}_{v,t} := 2m^{-0.2}\cdot  15^t +m\cdot \sum_{e\ni v; e \in E[V_i]} x_{e,t}^{MPC}  \geq \mathcal{T}_{v,t}\cdot w'(v)$:  freeze $v$ and its incident edges
             \label{algoline:inner freeze}
             \item For each active edge $e\in E[V_i]$: $x_{e,t+1}^{MPC} := x_{e,t}^{MPC}/(1-\epsilon)$
             \item For each frozen edge $e\in E[V_i]$: $x_{e,t+1}^{MPC} := x_{e,t}^{MPC}$
         \end{enumerate}
         \item For each $e=(u,v)\in E[V^{high}]$: $x_e^{MPC} \gets x_{e,0}^{MPC}/(1-\epsilon)^{t'}$, where $0\le t'<I$ is the earliest iteration in which either one of $u,v$ was frozen; or $t'=I$ if both remain active \label{algoline:update edges}
       
         \item For each active $v \in V^{high}$ satisfying $y_v^{MPC}:= \sum_{e\ni v; e\in E[V^{high}]}x_e^{MPC}\ge  w'(v)$: freeze $v$ and its incident edges \label{algoline:freeze v}
         \item For each $e\in E[V^{inactive};V^{high}]$: $x_e^{MPC} \gets 0$ \label{algoline:update inactive zero}
         \item Update residual degree for all nonfrozen $v$: $d(v) \gets $ number of nonfrozen neighbors of $v$ \label{algoline: update active d}
     \end{enumerate}
     \item Directly run the centralized algorithm in one machine on the subgraph induced by nonfrozen vertices, with residual weights $w'(v) \gets w(v) - \sum_{e \ni v, \text{ frozen}} x_e$ \label{algoline:central}
     \item Return all frozen vertices as a vertex cover
 \end{enumerate}
 \caption{MPC-Simulation for MWVC}
 \label{algo:mpc mwvc}
 \end{algorithm}

\subsection{MPC Simulation}
Our MPC algorithm is given in Algorithm \ref{algo:mpc mwvc}. As described in the previous section our MPC algorithm consists of several phases: each execution of the while-loop at Line~(\ref{algoline:phase}) is a phase, and the final execution of the centralized algorithm at Line~(\ref{algoline:central}) is the last phase of the algorithm.
Each phase consists of several iterations, which are similar as in the centralized algorithm.

Vertices may become frozen in a phase. Once frozen, they will remain frozen throughout the rest of the algorithm.
An edge is called frozen if and only if at least one of its endpoints is frozen.
If edge $e$ becomes frozen in a phase, then by the end of this phase---in particular, at Line~(\ref{algoline:update edges}) and Line~(\ref{algoline:update inactive zero})---it will be assigned a nonnegative weight $x_{e}^{MPC}$, which will never be changed in the following phases.
Since the weights of frozen edges are already finalized, we will use \emph{residual weight} $w'(v) = w(v) - \sum_{e \ni v\text{ frozen}} x_e^{MPC}$ as the weight of vertex $v$ when we start the new phase; this makes the analysis cleaner.
After the algorithm terminates, every edge $e\in E$ will be frozen and have a finalized edge weight $x_{e}^{MPC}$.


Per phase, we partition the high-degree vertices $V^{high}$ uniformly at random between $m=\sqrt{d}$ machines, and the respective induced subgraphs are gathered on each machine; the sampling probability is chosen such that the size of the induced subgraph for one machine does not exceed its memory constraint $\tilde{\bigo}(n)$.
Subsequently, the centralized algorithm is simulated in each of the $m$ induced subgraphs locally.
When comparing with the thresholds $\mathcal{T}_{v,t}$, we use $\tilde y_{v,t}^{MPC}$ as a local estimator of the total incident weight of $v$.

After all machines have finished their local simulation, we assign an edge weight $x_e^{MPC}$ to every edge $e=(u,v) \in E[V^{high}]$---especially those cross-partition edges which did not participate in the local simulation---based on the earliest iteration where either of $u,v$ became frozen during their respective local simulation. 
For frozen edges $e\in E[V^{high}]$, $x_{e}^{MPC}$ are their finalized edge weights. Edges in $E[V^{inactive};V^{high}]$ may have got frozen, too; their weights are finalized as 0. 
We also freeze the vertices whose sum of incident $x_{e}^{MPC}$ is too big, so as to prevent them from having negative residual weight in the next phase.
%
%
Once the average degree is below $\log^{30}n$, there are at most $n\log^{30}n \in {\widetilde{\bigo}}(n)$ edges left.\footnote{We did not attempt to optimize this 30 constant in the exponent.} Then, we move all edges into one machine, which executes the last iterations of the centralized algorithm.





\section{Analysis}
\label{chap:analysis}
In this chapter we provide an analysis of Algorithm \ref{algo:mpc mwvc}, the MPC simulation. The analysis is split into three major parts.
In Section~\ref{sec:memory} we address the memory constraints for the machines. 
In Section~\ref{sec:round comp} we derive the round complexity by analyzing the degree reduction in each phase.
Finally, we turn the attention to the approximation ratio in Section~\ref{section:approx ratio}.

\subsection{Memory Constraint}
\label{sec:memory}
Recall that $V^{high}$ is divided into subsets $V_1,\dots,V_m$, where each vertex $v\in V^{high}$ is independently randomly assigned to one of the subsets.
To simulate one phase, the $i$-th machine ($1\le i \le m$) needs to store the subgraph induced by $V_i$, together with edge weights $x_{e,0}^{MPC}$ and vertex weights $w'(v)$.
We do not need to store the random thresholds  $\mathcal{T}_{v,t}$, as they can be sampled on the fly. 
In the following lemma we show that the induced graph of $V_i$ with high probability contains at most $\bigo(n)$ edges, so the necessary information can fit into one machine. 
\begin{lemma}
At Line~(\ref{algoline:def vi}), with high probability $|E[V_i]| \in \bigo(n)$ hold for all $1\le i\le m$.
\end{lemma}

\begin{proof}

Recall that $m=\sqrt{d}$, and $V_1,\dots,V_m$ is a random partition of $V^{high}$. Fixing any $i \in \{1,\dots,m\}$, define independent random variables $s_v\in [0,1]$ for all $v\in V^{high}$ as: if $v\in V_i$ then $s_v:=d(v)/n$; otherwise $s_v:=0$.
Then $\E[\sum_{v\in V^{high}}s_v] = \sum_{v\in V^{high}} \frac{d(v)/n}{m} \le \frac{d}{m} = \sqrt{d}$. By Chernoff bound, we have 
$$\Pr[\sum_{v\in V^{high}}s_v > 2\sqrt{d}] \le \exp(-\sqrt{d}/3).$$
Recall that $d>\log^{30} n$. Hence, with high probability we have
$$\sum_{v\in V_i}d(v) = n\cdot \sum_{v\in V^{high}}s_v \le  2n\sqrt{d}.$$

For any $v\in V^{high}$, let $d_i(v)$ denote the number of its neighbors in $V_i$. Similarly by Chernoff bound, we have 
$$\Pr[d_i(v) >  \frac{2d(v)}{m}] \le \exp(-d(v)/3m) \le \exp(-d^{0.45}/3),$$
where the last inequality follows from the definition of $V^{high}$. Hence, with high probability we have
$$\lvert E(G'[V_i])\rvert= \frac{1}{2}\sum_{v \in V_i}d_i(v) \le \frac{1}{2}\sum_{v \in V_i}\frac{2d(v)}{m}  \le \frac{2n\sqrt{d}}{m} = 2n.$$
We finish the proof by a union bound over all $i\in \{1,\dots,m\}$.
\end{proof}

We have shown that the near-linear local memory constraint is satisfied.
Since the number of machines used for simulation is $m = \sqrt{d}\le \sqrt{|E|/n}$, the total memory used is with high probability $\tilde \bigo(\sqrt{d} n) \le \tilde \bigo(|E|)$, so the global memory constraint is also satisfied.


\subsection{Round Complexity}
\label{sec:round comp}
Our algorithm runs in multiple phases, each of which can be implemented in $O(1)$ MPC rounds.
The number of nonfrozen edges is reduced in each phase, and in the end we switch to the centralized algorithm when the nonfrozen edges fit in one machine.
We will bound the round complexity by showing that the number of nonfrozen edges significantly decreases in each phase.

As discussed in Section~\ref{sec3.2}, we use an orientation argument.
At the beginning of the phase, we orient the edges $e\in E[V^{high}]$ in the following way: direct the edge $(u,v)$ from $u$ to $v$ if $\frac{w'(u)}{d(u)} < \frac{w'(v)}{d(v)}$ and reverse otherwise, breaking ties arbitrarily.
After this orientation, the incident edges around vertex $v$ split in two parts: $N_{in}(v)$ contains all edges directed towards $v$ and $N_{out}(v)$ contains all edges directed outward from $v$. 
For each edge $e \in N_{out}(v)$ we have $x_{e,0}^{MPC} = \frac{w'(v)}{d(v)}$, and for each edge $e \in N_{in}(v)$ we have $x_{e,0}^{MPC} \le \frac{w'(v)}{d(v)}$. 


\begin{remark}
Note that $d(v)$ is defined as the number of nonfrozen neighbors of $v$ (see Line~(\ref{algoline: update active d})), i.e., the degree of $v$ in the subgraph induced by $V^{high}\cup V^{inactive}$. It is \emph{not} defined as the number of $v$'s neighbors in $V^{high}$.
\label{remark:deg}
\end{remark}

A first observation to make is that the active out-degree decreases significantly over a phase.
\begin{observation}[Active out-degree]
\label{obs:act out degree}
After Line~(\ref{algoline:freeze v}) finishes, for any active vertex $v\in V^{high}$, denote with $d_A^{out}(v)$ the number of edges $(v,u)$ directed outward from $v$ such that $u\in V^{high}$ is still active. 
 Then 
$$
d^{out}_A(v) \leq d(v)(1-\epsilon)^I.
$$
\end{observation}

\begin{proof}
Recall that for every edge $e\in E[V^{high}]$ directed outward from $v$, we set $x_{e,0}^{MPC}= \frac{w'(v)}{d(v)}$ at the beginning of this phase.
Assume towards a contradiction that after Line~(\ref{algoline:freeze v}) finishes there exists an active $v\in V^{high}$ with $d^{out}_A(v) > d(v)(1-\epsilon)^I$.
The weight of active out-edges is at this point $x_{e}^{MPC} =x_{e,I}^{MPC} = \frac{w'(v)}{d(v)(1-\epsilon)^I}$. Therefore, \begin{align*}
y_v^{MPC} &  = \sum_{e\ni v;e\in E[V^{high}]}x_{e}^{MPC}\\
& \geq  d_A^{out}(v)\cdot \frac{w'(v)}{d(v)(1-\epsilon)^I} \\
 & >  w'(v),
\end{align*}
meaning that $v$ would have been frozen at Line~(\ref{algoline:freeze v}), hence there is no such vertex.
\end{proof}
Using Observation \ref{obs:act out degree} we show that the number of nonfrozen edges is decreasing over the course of a phase.
\begin{lemma}
With high probability, after Line~(\ref{algoline: update active d}) finishes, the number of remaining nonfrozen edges is 
$$\frac{1}{2}\sum_{v\in V\text{ nonfrozen}} d(v) \le 2nd(1-\eps)^{I}.$$
\label{lemma:decrease avg degree}
\end{lemma}
\begin{proof}
A remaining nonfrozen edge is either incident to a low-degree vertex $v\in V^{inactive}$, or in $E[V^{high}]$ and has been active during this phase.
Recall that $$nd = \sum_{V^{high}\cup V^{inactive}}d(v)$$ (see Line~(\ref{algoline:phase}) and Line~(\ref{algoline:alive vert})). Combining Observation~\ref{obs:act out degree} and the definition of $V^{inactive}$, we have that the number of remaining active edges is at most
\begin{align*}
\sum_{v\in V^{high}} d^{out}_A(v) + n\cdot d^{0.95} &\leq \sum_{v\in V^{high}} d(v)(1-\epsilon)^I + nd^{0.95} \\ &\leq nd(1-\epsilon)^I + nd^{0.95} \\
&\leq 2nd(1-\epsilon)^I.
\end{align*}
The last inequality holds as
$(1-\epsilon)^I = (1-\epsilon)^{\log d/20\log 15} \ge d^{-1/20}$.
\end{proof}

\begin{theorem}
The MPC simulation with high probability takes at most ${\bigo}(\log\log d)$ rounds, where $d= 2|E|/n$ is the initial average degree of the input graph.
\end{theorem}

\begin{proof}
As each phase takes $O(1)$ MPC rounds, we now need to show that the condition at Line~(\ref{algoline:phase}) breaks after $\bigo(\log \log d)$ phases. 
Using Lemma~\ref{lemma:decrease avg degree} and the definition of $I$ (see Line~(\ref{algoline:nb machines})), it remains to show that $d_k\le \log^{30} n$ holds for some $k\in \bigo(\log \log d)$, where $d_0=d$ and $d_{i+1} \le 4d_i(1-\eps)^{\frac{\log d_i}{20\log 15}}$.

Assume $0<\eps <1/2$ and define the constant $\gamma = \frac{\log(1/(1-\eps))}{40\log 15}\in (0,1)$. When $d_i>\log^{30} n$, we have 
$$d_{i+1} \le 4d_i^{1-2\gamma} \le d_i^{1-\gamma}$$
for sufficiently large $n$. 
Hence we have $d_k\le d_{k-1}^{1-\gamma}\le \dots \le d^{(1-\gamma)^k} \le \log^{30} n$ for some $k \le \frac{\log \big (\log d\big /30\log\log n\big )}{\log\big (1\big /(1-\gamma) \big )}\in \bigo (\log \log d)$.
\end{proof}

\subsection{Approximation Ratio}
\label{section:approx ratio}
In this section, we will prove the approximation guarantee of our MPC algorithm.
To achieve this, we consider one phase of the MPC algorithm, and imagine running the centralized algorithm on the induced subgraph of $V^{high}$, with the same vertex weights $w'(v)$, initial edge weights $x_{e,0}^{MPC}$ and random thresholds $\mathcal{T}_{v,t}$ that are used in this phase of MPC simulation, and then compare the behaviour of these two algorithms.

Recall that the edge weights and total incident weights in the centralized algorithm are denoted by $x_{e,t}$ and $y_{v,t}$.
We have $x_{e,0}:=x_{e,0}^{MPC}$.
By Proposition~\ref{prop:central initialization} we know this initialization is valid.

In the description of the MPC algorithm, we have defined edge weights $x_{e,t}^{MPC}$ $(0\le t\le I)$ for all local edges $e\in E[V_1]\cup \dots \cup E[V_m]$.
We can easily extend this definition to all edges in $E[V^{high}]$ (that is, including cross-partition edges), in the sense as given by Line~(\ref{algoline:update edges}). 
Similarly we can define $y_{v,t}^{MPC}$ for all $v\in V^{high}, 0\le t \le I$ (as at Line~(\ref{algoline:freeze v})).

We will prove the following key lemma, stating that the behaviour of our MPC simulation is similar to that of the centralized algorithm.
\begin{lemma}
\label{lemma:compare}
Consider one phase of the MPC algorithm.
Run the centralized algorithm for $I:=\frac{\log m}{10\log 15}$ iterations on the graph induced by $V^{high}$, with the same vertex weights $w'$, initial edge weights $x_{e,0}^{MPC}$, and random thresholds $\mathcal{T}_{v,t}$.
With high probability, we have
$$|y_{v,t}-\tilde y_{v,t}^{MPC}| \le 6\eps \cdot w'(v),$$
and
$$|y_{v,t}-y_{v,t}^{MPC}| \le 6\eps \cdot w'(v),$$
for all $0\le t\le  I$ and all $v\in V^{high}$.
\end{lemma}

Before proving this lemma, we show how it implies the approximation guarantee of our MPC algorithm.
\begin{theorem}
With high probability, Algorithm~\ref{algo:mpc mwvc} returns a vertex cover $C$ which satisfies 
$ w(C) \leq (2+30\epsilon) OPT,$
where $OPT$ is the weight of a minimum weight vertex cover.
\end{theorem}
\begin{proof}
The returned vertex cover consists of the vertices that were frozen in our MPC algorithm.
When vertex $v$ is frozen at Line~(\ref{algoline:inner freeze}) in some phase, we have $\tilde y^{MPC}_{v,t}\ge \mathcal{T}_{v,t}w'(v) \ge (1-4\epsilon) w'(v)$.
By Lemma~\ref{lemma:compare}, this implies $y_{v,t}^{MPC} \ge (1-16\epsilon)w'(v)$.
Recall that $w(v)-w'(v)\ge 0$ is the total weight of $v$'s incident edges that were frozen in previous phases.
So we have $$\sum_{e\ni v} x_e^{MPC} = w(v)-w'(v) + y_{v,t}^{MPC}  \ge (1-16\epsilon)w(v).$$
Note that this inequality holds as well for those vertices $v$ that became frozen at Line~(\ref{algoline:freeze v}) or in the final centralized phase (Line~(\ref{algoline:central})). 

On the other hand,  by Observation~\ref{obs:central valid}, $y_{u,t} \le w'(u)$ always holds for all vertices $u\in V^{high}$. Then by Lemma~\ref{lemma:compare}, we have $y_{u,t}^{MPC} \le (1+6\eps)w'(u)$ with high probability. 
Similarly, this implies
\[\sum_{e\ni u} x_e^{MPC} \le (1+6\epsilon)w(u)\]
in the end.

Hence, $\{x_e^{MPC}/(1+6\eps)\}_{e\in E}$ is a valid fractional matching, and by the same argument as in the proof of Proposition~\ref{prop:central apx}, we have 
\[w(C)\le 2 \sum_{e\ni v}\frac{x^{MPC}_e}{1-16\eps}\le \frac{2(1+6\eps)OPT}{1-16\eps} \le (2+30\eps)OPT. \qedhere\]
\end{proof}
The rest of this section is devoted to proving Lemma~\ref{lemma:compare}.

In a phase, we say that vertex $v$ becomes \emph{bad} if it gets frozen in the centralized algorithm 
and not in MPC simulation (or the other way around).
Once bad, the vertex remains bad throughout the whole phase.
If a vertex is not bad, we say it is \emph{good}.

In order to bound the weight of bad vertices, we will show that the estimated values $\tilde{y}^{MPC}_{v,t}$ and $y_{v,t}^{MPC}$ stay close to the actual values $y_{v,t}$ during one phase of MPC simulation.
Once we establish that, we can use Lemma \ref{lemma:probability bad} below to bound the total weight of adjacent vertices which turn bad in a particular iteration.

\begin{lemma}
Suppose $|y_{v,t} - \tilde{y}^{MPC}_{v,t}| \leq \sigma w'(v)$ holds for all vertices $v$ that are active in both the centralized algorithm and MPC simulation.
Then, $v$ becomes bad in iteration $t$ with probability at most $\sigma/\epsilon$ and independently of other vertices.
\label{lemma:probability bad}
\end{lemma}

\begin{proof}
For a vertex $v$ to become bad, the estimate $\tilde{y}^{MPC}_{v,t}$ has to be on the other side of the threshold as $y_{v,t}$.
Call the effective threshold $T_{v,t} = \mathcal{T}_{v,t} w'(v)$.
Notice that when $|\tilde{y}^{MPC}_{v,t} - T_{v,t}| > \sigma w'(v)$ it is not possible for vertex $v$ to become bad.
Therefore, only if $T_{v,t}$ falls in the interval of size $2\sigma w'(v)$ around $\tilde{y}^{MPC}_{v,t}$, $v$ might become bad. $\mathcal{T}_{v,t}$ is chosen uniformly at random from an interval with size $2\epsilon$, this is equivalent to pick $T_{v,t}$ uniformly at random from the scaled interval with size $2\epsilon w'(v)$.
Therefore the possibility of $v$ becoming bad in iteration $t$ is upper-bounded by $2\sigma w'(v)/(2\epsilon w'(v)) = \sigma/\epsilon$.
\end{proof}

In order to compare $y_{v,t}$ and $\tilde y_{v,t}^{MPC}$, we introduce an intermediate quantity defined as follows.
\begin{definition}
\label{defini:intermediate}
For $v\in V^{high}$, let \[\tilde y_{v,t}  := 2m^{-0.2}\cdot  15^t +m\cdot \sum_{e\ni v; e \in E[V_i]}x_{e,t},\]
where $v \in V_i$. Note that $\tilde y_{v,0} = \tilde y^{MPC}_{v,0}$, since $x_{e,0} = x^{MPC}_{e,0}$.
\end{definition}

Instead of directly calculating a bound on $|y_{v,t} - y_{v,t}^{MPC}|$ and $|\tilde y_{v,t}-\tilde y_{v,t}^{MPC}|$, we use a slightly different notion introduced below.

\begin{definition}[Weight difference] 
Let
$$
\diff(v,t) := \sum_{e\ni v; e \in E[V^{high}]} |x_{e,t} - x_{e,t}^{MPC}|.
$$
And, let
$$
\diff^{local}(v,t) := m\cdot \sum_{e\ni v; e\in E[V_i]} |x_{e,t} - x_{e,t}^{MPC}|,
$$
where $v\in V_i$.
\\
Note that $|y_{v,t} - y_{v,t}^{MPC}| \leq \diff(v,t)$, and $|\tilde y_{v,t}- \tilde y^{MPC}_{v,t}| \le \diff^{local}(v,t)$.
\end{definition}

Now we prove the concentration lemma, which immediately implies a bound on $\tilde y_{v,t}-y_{v,t}$.
\begin{lemma}[Concentration]
 Let $1\le r \le m^{1.2}$, and let $U\subseteq V^{high}$ be a random subset where each vertex is included with probability $1/r$ independently.
 For any $0\le t\le  I$ and any vertex $v \in V^{high}$, with high probability,
$$
\Big \lvert y_{v, t} - r\sum_{u\in U; (v,u)\in E[V^{high}]}x_{(v,u),t}\Big \rvert < m^{-0.2} \cdot w'(v).
$$
\label{lemma:concentration}
\end{lemma}

\begin{proof}
Denote $w_{out}(v) = \frac{w'(v)}{d(v)}\big / (1-\eps)^t$.
For every incident edge $e=(u,v)$, we have \[ \textstyle 0\le x_{e,t} \le x_{e,0}/(1-\eps)^t= \min\left \{\frac{w'(u)}{d(u)},\frac{w'(v)}{d(v)}\right \}\big / (1-\eps)^t\le w_{out}(v).\]
Let \[ \textstyle X_v := \sum_{u \in U: (u,v)\in E[V^{high}]} x_{(u,v),t}/w_{out}(v),\] which is the sum of independent random variables in interval $[0,1]$. 
By the definition of set $U$ it is clear that \[ \textstyle \E[X_v] = \frac{1}{r}\sum_{e\ni v; e\in E[V^{high}]} x_{e,t}/w_{out}(v).\]
By Observation~\ref{obs:central valid}, $\sum_{e\ni v; e\in E[V^{high}]} x_{e,t} \le w'(v) \le d(v)w_{out}(v)$. We then obtain 
\begin{equation}
    d(v) \ge r\E[X_v].\label{eqn:temp}
\end{equation}
Let $$\delta:= \frac{m^{0.65}\sqrt{d(v)}}{r \E[X_v]}.$$ We use Chernoff bound and analyze two cases:
\begin{itemize}
    \item If $\delta<1$, 
    \begin{align*}
      \Pr\big [|X_v - \E[X_v]| \ge \delta \E[X_v]\big ] & \le 2\exp(-\delta^2\E[X_v]/3) \\
       & = 2\exp\left (-\frac{m^{1.3}d(v)}{3r^2\E[X_v]}\right )\\
       &\le 2\exp(-m^{0.1}/3),
    \end{align*}
    where the last inequality follows from (\ref{eqn:temp}) and $r\le m^{1.2}$.
    \item If $\delta \ge 1$, 
    \begin{align*}
       \Pr\big [|X_v - \E[X_v]| \ge \delta \E[X_v]\big ] & \le 2\exp(-\delta\E[X_v]/3) \\ 
       & = 2\exp\left (-\frac{m^{0.65}\sqrt{d(v)}}{3r}\right )\\
       & \le 2\exp(-m^{0.4}/3),
    \end{align*}
    where the last inequality follows from $d(v) \ge d^{0.95}=m^{1.9}$ for all $v\in V^{high}$ and $r\le m^{1.2}$.
\end{itemize}
As $m = \sqrt{d}>\log^{15}n$, we conclude that with high probability $|X_v - \E[X_v]| < \delta \E[X_v]$, or equivalently
\begin{align*}
\textstyle \lvert y_{v, t} - r\sum_{u\in U; (v,u)\in E[V^{high}]}x_{(v,u),t} \rvert
 &< r w_{out}(v)\cdot \delta \E[X_v] \\
&= \frac{m^{0.65}}{\sqrt{d(v)}}\cdot w'(v)/(1-\eps)^t\\
&\le m^{-0.3}\cdot w'(v)/(1-\eps)^t\\
&\le m^{-0.2}\cdot w'(v),
\end{align*}
where the last inequality follows from $(1/(1-\epsilon))^t \le (1/(1-\epsilon))^I \leq m^{0.1}$.
\end{proof}

\begin{corollary}
For any $0\le t\le  I$ and any vertex $v\in V^{high}$, with high probability,
\[(2\cdot 15^t-1)m^{-0.2} w'(v)\le \tilde{y}_{v,t} - y_{v, t} \le (2\cdot 15^t+1)m^{-0.2} w'(v).\]
\label{cor:initial estimate}
\end{corollary}
\begin{proof}
The proof directly follows from the definition of $V_i$, Definition~\ref{defini:intermediate}, and Lemma~\ref{lemma:concentration}.
\end{proof}
Now we will use induction to show that, $y_{v,t},\tilde y_{v,t},y_{v,t}^{MPC}$, and $\tilde y_{v,t}^{MPC}$ stay close to each other, for every good vertex $v$.

\begin{lemma}
\label{lemma:evolution}
For every $0\le t\le I$, the following hold with high probability for every $v\in V^{high}$ that is good by the start of iteration $t$:
\begin{enumerate}
 \setlength{\itemsep}{2pt}
\item 
 $\diff(v,t) \le m^{-0.2}\cdot 15^t \cdot w'(v)$,
 \item $\diff^{local}(v,t) \le m^{-0.2}\cdot 15^t \cdot w'(v)$, and,
 \item 
  $0\le \tilde y_{v,t}^{MPC} -y_{v,t} \le 4\cdot m^{-0.2}\cdot 15^t \cdot w'(v).$
\end{enumerate}
\end{lemma}
\begin{proof}
Note that Item~(3) directly follows from Item~(2) and Corollary~\ref{cor:initial estimate}:
\begin{align*}
      \tilde y^{MPC}_{v,t} - y_{v,t} &=  (\tilde y^{MPC}_{v,t} - \tilde y_{v,t}  ) + (\tilde y_{v,t} -  y_{v,t})\\
    & \le \diff^{local}(v,t) + m^{-0.2}(2\cdot 15^t+1) w'(v)   \\
    & \le m^{-0.2}\cdot 4\cdot 15^t\cdot w'(v),
\end{align*}
and,
\begin{align*}
      \tilde y^{MPC}_{v,t} - y_{v,t} &=  (\tilde y^{MPC}_{v,t} - \tilde y_{v,t}  ) + (\tilde y_{v,t} -  y_{v,t})\\
    & \ge -\diff^{local}(v,t) + m^{-0.2}(2\cdot 15^t-1) w'(v)   \\
    & \ge 0.
\end{align*}

In the following we will prove Item~(1) and Item~(2) by induction. 
Assuming the statements hold for some $t\ge 0$, we bound $\diff(v,t+1)$ and $\diff^{local}(v,t+1)$ by analyzing the evolution of weight differences in iteration $t$, for all $v$ that remain good after iteration $t$ finishes.

\subparagraph*{\textbf{Old bad vertices:}} If $|x_{e,t}-x_{e,t}^{MPC}|>0$ for some $e=(v,u)$, then we must have $x_{e,t} = x_{e,0}/(1-\eps)^{t_1}$, $x_{e,t}^{MPC} = x_{e,0}/(1-\eps)^{t_2}$ for some $t_1\neq t_2$, that is, $e$ was frozen at different iterations in the centralized and the MPC algorithms, or $e$ was frozen in one algorithm but is still active in the other . In the former case, we have  $|x_{e,t}-x_{e,t}^{MPC}| = |x_{e,t+1}-x_{e,t+1}^{MPC}|$. In the latter case, we assume w.l.o.g. $t_1<t_2 = t$, and then
\begin{align*}
    \frac{|x_{e,t+1}-x_{e,t+1}^{MPC}|}{|x_{e,t}-x_{e,t}^{MPC}|}  &=
    \frac{|x_{e,0}/(1-\eps)^{t_1}-x_{e,0}/(1-\eps)^{t_2+1}|}{|x_{e,0}/(1-\eps)^{t_1}-x_{e,0}/(1-\eps)^{t_2}|}\\
    &= \frac{1-(1-\eps)^{1+t_2-t_1}}{1-(1-\eps)^{t_2-t_1}}\\
    & < 3,
\end{align*}
for small enough $\eps$.

\subparagraph*{\textbf{New bad vertices:}}
We now analyze the edges $e=(v,u)$ such that $|x_{e,t}-x_{e,t}^{MPC}|=0$.
If $|x_{e,t+1}-x_{e,t+1}^{MPC}|>0$, then it must be that $u$ turns bad in iteration $t$ (since $v$ remains good), and we have \[|x_{e,t+1}-x_{e,t+1}^{MPC}| = x_{e,t}/(1-\eps)-x_{e,t} \le 2\eps x_{e,t}.\]

By Item~(3) and Lemma~\ref{lemma:probability bad}, each $u$ that was good by the start of iteration $t$ turns bad with probability at most \[4m^{-0.2}\cdot 15^t\cdot w'(v)/\eps\] independently.
Hence, by Lemma~\ref{lemma:concentration}, the total contribution of such $|x_{e,t+1}-x_{e,t+1}^{MPC}|$ in $\diff(v,t+1)$ is with high probability at most 
\begin{align*}
    \sum_{ \begin{smallmatrix}{{\scriptstyle e=(v,u)\in E[V^{high}];}}\\ \scriptstyle\text{  $u$ turns bad}\end{smallmatrix}}2\eps x_{e,t} 
    &\le 2\eps \cdot \frac{4m^{-0.2}\cdot 15^t w'(v)}{\eps} ( y_{v,t} + m^{-0.2}w'(v))\\
    & \le 12 m^{-0.2}\cdot 15^t \cdot  w'(v),
\end{align*}
where we used $y_{v,t}\le w'(v)$ and assumed $m^{-0.2}\le 1/2$.
Similarly, their contribution in $\diff^{local}(v,t+1)$ is with high probability at most
\begin{align*}
    m\cdot \sum_{ \begin{smallmatrix}{{\scriptstyle e=(v,u)\in E[V_{i}];}}\\ \scriptstyle\text{  $u$ turns bad}\end{smallmatrix}}2\eps x_{e,t} 
    &\le m\cdot 2\eps \frac{4m^{-0.2} 15^t w'(v)}{m\eps} ( y_{v,t} + m^{-0.2}w'(v))\\
    & \le 12 m^{-0.2}\cdot 15^t \cdot  w'(v).\\
\end{align*}

Finally, combining the effect of old bad vertices and new bad vertices, we have
\begin{align*}
  & \diff(v,t+1) \\ = \ &\sum_{e\ni v;e\in E[V^{high}]}|x_{e,t+1}-x_{e,t+1}^{MPC}|\\
   =\ & \sum_{\begin{smallmatrix}{{\scriptstyle e\ni v; e \in E[V^{high}];}}\\ \scriptstyle |x_{e,t}-x_{e,t}^{MPC}|>0\end{smallmatrix}}|x_{e,t+1}-x_{e,t+1}^{MPC}| \,  + \sum_{\begin{smallmatrix}{{\scriptstyle e\ni v; e \in E[V^{high}];}}\\ \scriptstyle |x_{e,t}-x_{e,t}^{MPC}|=0\end{smallmatrix}}|x_{e,t+1}-x_{e,t+1}^{MPC}| \\
   \le \ & \sum_{\begin{smallmatrix}{{\scriptstyle e\ni v; e \in E[V^{high}];}}\\ \scriptstyle |x_{e,t}-x_{e,t}^{MPC}|>0\end{smallmatrix}}3|x_{e,t}-x_{e,t}^{MPC}| \,  + \sum_{\begin{smallmatrix}{{\scriptstyle e=(v,u)\in E[V^{high}];}}\\ \scriptstyle \text{$u$ turns bad }\end{smallmatrix}}2\eps x_{e,t} \\
   \le \ & 3\diff(v,t) + 12 m^{-0.2}\cdot 15^t \cdot  w'(v)\\
   \le \ & m^{-0.2}\cdot 15^{t+1}\cdot w'(v).
\end{align*}
Similarly we can show 
$\diff^{local}(v,t+1)\le m^{-0.2}\cdot 15^{t+1} \cdot w'(v)$.
\end{proof}

\begin{reminder}{Lemma~\ref{lemma:compare}}
With high probability, we have
$$|y_{v,t}-\tilde y_{v,t}^{MPC}| \le 6\eps \cdot w'(v),$$
and
$$|y_{v,t}-y_{v,t}^{MPC}| \le 6\eps \cdot w'(v),$$
for all $0\le t\le  I$ and all $v\in V^{high}$.
\end{reminder}

\begin{proof}
If $v$ is good by the start of iteration $t$, by Lemma~\ref{lemma:evolution},
\[ 0\le \tilde y_{v,t}^{MPC} - y_{v,t} \le 4\cdot 15^{t} m^{-0.2} w'(v) \le  4m^{-0.1}w'(v)\le \eps w'(v),\]
where we used $t\le I = \log m/(10\log 15)$, and assumed $4m^{-0.1}\le \eps$.
And similarly we have 
\[|y_{v,t}-y_{v,t}^{MPC}|\le \diff(v,t)\le  15^{t} m^{-0.2} w'(v) \le \eps w'(v).\]

Otherwise, suppose $v$ turned bad in iteration $t^*<t$.
Because $\tilde y_{v,t^*}^{MPC} - y_{v,t^*}\ge 0$, it must be that $v$ became frozen in MPC simulation while remained active in the centralized algorithm in iteration $t^*$. So
\[ \tilde y_{v,t}^{MPC}  = \tilde y_{v,t^*}^{MPC} \ge \mathcal{T}_{v,t^*}w'(v)\ge (1-4\eps)w'(v),\]
Since $v$ was good by the start of iteration $t^*$, we have 
\[y_{v,t^*} \ge \tilde y_{v,t^*}^{MPC}-\eps w'(v) \ge (1-5\eps)w'(v). \] 
Then by monotonicity of $x_{e,t}$ and Observation~\ref{obs:central valid}, we have 
$y_{v,t^{*}}\le y_{v,t} \le w'(v)$,
implying that \[|y_{v,t^*}-y_{v,t}| \le 5\eps w'(v).\]

Hence, we have 
\begin{align*}
|y_{v,t}-\tilde y_{v,t}^{MPC}|&=|y_{v,t}-\tilde y_{v,t^*}^{MPC}|\\
& \le |y_{v,t}-y_{v,t^*}|+|y_{v,t^*}-\tilde y_{v,t^*}^{MPC}|\\
& \le 5\eps w'(v)+ \eps w'(v)\\
& = 6\eps w'(v).
\end{align*}
Similarly,
\[|y_{v,t}-y_{v,t}^{MPC}| \le 6\eps w'(v). \qedhere\]
\end{proof}

\section*{Acknowledgements}
We are grateful to the reviewers of SPAA 2020 for their helpful comments.

The first author's work in this project was supported by funding from the European Research Council (ERC), under the European Union's Horizon 2020 research and innovation programme (grant agreement No. 853109).

	\bibliographystyle{alphaurl} 
	\bibliography{refs}

\newcommand{\etalchar}[1]{$^{#1}$}
\begin{thebibliography}{LPSPP05}

\bibitem[ABB{\etalchar{+}}19]{Assadi2017CoresetsME}
Sepehr Assadi, MohammadHossein Bateni, Aaron Bernstein, Vahab Mirrokni, and
  Cliff Stein.
\newblock Coresets meet {EDCS}: Algorithms for matching and vertex cover on
  massive graphs.
\newblock In {\em Proceedings of the 30th {ACM-SIAM} Symposium on Discrete
  Algorithms ({SODA})}, pages 1616--1635, 2019.
\newblock \href {http://dx.doi.org/10.1137/1.9781611975482.98}
  {\path{doi:10.1137/1.9781611975482.98}}.

\bibitem[ANOY14]{andoni2014parallel}
Alexandr Andoni, Aleksandar Nikolov, Krzysztof Onak, and Grigory Yaroslavtsev.
\newblock Parallel algorithms for geometric graph problems.
\newblock In {\em Proceedings of the 46th {ACM} Symposium on Theory of
  Computing ({STOC})}, pages 574--583, 2014.
\newblock \href {http://dx.doi.org/10.1145/2591796.2591805}
  {\path{doi:10.1145/2591796.2591805}}.

\bibitem[BBD{\etalchar{+}}19]{behnezhad2019massively}
Soheil Behnezhad, Sebastian Brandt, Mahsa Derakhshan, Manuela Fischer,
  MohammadTaghi Hajiaghayi, Richard~M. Karp, and Jara Uitto.
\newblock Massively parallel computation of matching and {MIS} in sparse
  graphs.
\newblock In {\em Proceedings of the 38th {ACM} Symposium on Principles of
  Distributed Computing ({PODC})}, pages 481--490, 2019.
\newblock \href {http://dx.doi.org/10.1145/3293611.3331609}
  {\path{doi:10.1145/3293611.3331609}}.

\bibitem[BDH18]{Behnezhad2018BriefAS}
Soheil Behnezhad, Mahsa Derakhshan, and MohammadTaghi Hajiaghayi.
\newblock Semi-{MapReduce} meets congested clique.
\newblock {\em arXiv preprint}, 1802.10297, 2018.
\newblock \href {http://arxiv.org/abs/1802.10297} {\path{arXiv:1802.10297}}.

\bibitem[BHH19]{behnezhad2019exponentially}
Soheil Behnezhad, MohammadTaghi Hajiaghayi, and David~G. Harris.
\newblock Exponentially faster massively parallel maximal matching.
\newblock In {\em Proceedings of the 60th {IEEE} Symposium on Foundations of
  Computer Science ({FOCS})}, pages 1637--1649, 2019.
\newblock \href {http://dx.doi.org/10.1109/FOCS.2019.00096}
  {\path{doi:10.1109/FOCS.2019.00096}}.

\bibitem[BKS17]{beame2017communication}
Paul Beame, Paraschos Koutris, and Dan Suciu.
\newblock Communication steps for parallel query processing.
\newblock {\em Journal of the ACM}, 64(6):1--58, 2017.
\newblock \href {http://dx.doi.org/10.1145/3125644}
  {\path{doi:10.1145/3125644}}.

\bibitem[BYE81]{bar1981linear}
Reuven Bar-Yehuda and Shimon Even.
\newblock A linear-time approximation algorithm for the weighted vertex cover
  problem.
\newblock {\em Journal of Algorithms}, 2(2):198--203, 1981.
\newblock \href {http://dx.doi.org/10.1016/0196-6774(81)90020-1}
  {\path{doi:10.1016/0196-6774(81)90020-1}}.

\bibitem[C{\L}M{\etalchar{+}}19]{Czumaj2018RoundCF}
Artur Czumaj, Jakub {\L}{\k{a}}cki, Aleksander M{\k{a}}dry, Slobodan
  Mitrovi{\'{c}}, Krzysztof Onak, and Piotr Sankowski.
\newblock Round compression for parallel matching algorithms.
\newblock {\em {SIAM} Journal on Computing}, pages STOC18--1--STOC18--44, 2019.
\newblock \href {http://dx.doi.org/10.1137/18M1197655}
  {\path{doi:10.1137/18M1197655}}.

\bibitem[DG08]{Dean2004MapReduceSD}
Jeffrey Dean and Sanjay Ghemawat.
\newblock {MapReduce}: Simplified data processing on large clusters.
\newblock {\em Communications of the {ACM}}, 51(1):107--113, 2008.
\newblock \href {http://dx.doi.org/10.1145/1327452.1327492}
  {\path{doi:10.1145/1327452.1327492}}.

\bibitem[DP09]{Dubhashi2009ConcentrationOM}
Devdatt Dubhashi and Alessandro Panconesi.
\newblock {\em Concentration of measure for the analysis of randomized
  algorithms}.
\newblock Cambridge University Press, 2009.

\bibitem[FMS{\etalchar{+}}10]{Feldman2008OnDS}
Jon Feldman, Shanmugavelayutham Muthukrishnan, Anastasios Sidiropoulos, Cliff
  Stein, and Zoya Svitkina.
\newblock On distributing symmetric streaming computations.
\newblock {\em ACM Transactions on Algorithms}, 6(4):1--19, 2010.
\newblock \href {http://dx.doi.org/10.1145/1824777.1824786}
  {\path{doi:10.1145/1824777.1824786}}.

\bibitem[GGK{\etalchar{+}}18]{Ghaffari2018ImprovedMP}
Mohsen Ghaffari, Themis Gouleakis, Christian Konrad, Slobodan Mitrovi{\'c}, and
  Ronitt Rubinfeld.
\newblock Improved massively parallel computation algorithms for {MIS},
  matching, and vertex cover.
\newblock In {\em Proceedings of the 37th {ACM} Symposium on Principles of
  Distributed Computing ({PODC})}, pages 129--138, 2018.
\newblock \href {http://dx.doi.org/10.1145/3212734.3212743}
  {\path{doi:10.1145/3212734.3212743}}.

\bibitem[GKMS19]{gamlath2019weighted}
Buddhima Gamlath, Sagar Kale, Slobodan Mitrovi{\'c}, and Ola Svensson.
\newblock Weighted matchings via unweighted augmentations.
\newblock In {\em Proceedings of the 38th {ACM} Symposium on Principles of
  Distributed Computing ({PODC})}, pages 491--500, 2019.
\newblock \href {http://dx.doi.org/10.1145/3293611.3331603}
  {\path{doi:10.1145/3293611.3331603}}.

\bibitem[GSZ11]{Goodrich2011SortingSA}
Michael~T. Goodrich, Nodari Sitchinava, and Qin Zhang.
\newblock Sorting, searching, and simulation in the {MapReduce} framework.
\newblock In {\em Proceedings of the 22nd International Symposium on Algorithms
  and Computation ({ISAAC})}, pages 374--383, 2011.
\newblock \href {http://dx.doi.org/10.1007/978-3-642-25591-5_39}
  {\path{doi:10.1007/978-3-642-25591-5_39}}.

\bibitem[GU19]{Ghaffari2018SparsifyingDA}
Mohsen Ghaffari and Jara Uitto.
\newblock Sparsifying distributed algorithms with ramifications in massively
  parallel computation and centralized local computation.
\newblock In {\em Proceedings of the 30th {ACM-SIAM} Symposium on Discrete
  Algorithms ({SODA})}, pages 1636--1653, 2019.
\newblock \href {http://dx.doi.org/10.1137/1.9781611975482.99}
  {\path{doi:10.1137/1.9781611975482.99}}.

\bibitem[Hoc82]{hochbaum1982approximation}
Dorit~S. Hochbaum.
\newblock Approximation algorithms for the set covering and vertex cover
  problems.
\newblock {\em {SIAM} Journal on Computing}, 11(3):555--556, 1982.
\newblock \href {http://dx.doi.org/10.1137/0211045}
  {\path{doi:10.1137/0211045}}.

\bibitem[IBY{\etalchar{+}}07]{Isard2007DryadDD}
Michael Isard, Mihai Budiu, Yuan Yu, Andrew Birrell, and Dennis Fetterly.
\newblock Dryad: Distributed data-parallel programs from sequential building
  blocks.
\newblock In {\em Proceedings of the 2nd ACM SIGOPS/EuroSys European Conference
  on Computer Systems}, pages 59--72, 2007.
\newblock \href {http://dx.doi.org/10.1145/1272996.1273005}
  {\path{doi:10.1145/1272996.1273005}}.

\bibitem[II86]{israeli1986fast}
Amos Israeli and Alon Itai.
\newblock A fast and simple randomized parallel algorithm for maximal matching.
\newblock {\em Information Processing Letters}, 22(2):77--80, 1986.
\newblock \href {http://dx.doi.org/10.1016/0020-0190(86)90144-4}
  {\path{doi:10.1016/0020-0190(86)90144-4}}.

\bibitem[KSV10]{Karloff2010AMO}
Howard Karloff, Siddharth Suri, and Sergei Vassilvitskii.
\newblock A model of computation for {MapReduce}.
\newblock In {\em Proceedings of the 21st {ACM-SIAM} Symposium on Discrete
  Algorithms ({SODA})}, pages 938--948, 2010.
\newblock \href {http://dx.doi.org/10.1137/1.9781611973075.76}
  {\path{doi:10.1137/1.9781611973075.76}}.

\bibitem[KY09]{koufogiannakis2009distributed}
Christos Koufogiannakis and Neal~E. Young.
\newblock Distributed and parallel algorithms for weighted vertex cover and
  other covering problems.
\newblock In {\em Proceedings of the 28th {ACM} Symposium on Principles of
  Distributed Computing ({PODC})}, pages 171--179, 2009.
\newblock \href {http://dx.doi.org/10.1145/1582716.1582746}
  {\path{doi:10.1145/1582716.1582746}}.

\bibitem[Lin92]{Linial1992LocalityID}
Nathan Linial.
\newblock Locality in distributed graph algorithms.
\newblock {\em {SIAM} Journal on Computing}, 21(1):193--201, 1992.
\newblock \href {http://dx.doi.org/10.1137/0221015}
  {\path{doi:10.1137/0221015}}.

\bibitem[LMSV11]{Lattanzi2011FilteringAM}
Silvio Lattanzi, Benjamin Moseley, Siddharth Suri, and Sergei Vassilvitskii.
\newblock Filtering: A method for solving graph problems in {MapReduce}.
\newblock In {\em Proceedings of the 23rd {ACM} Symposium on Parallelism in
  Algorithms and Architectures ({SPAA})}, pages 85--94, 2011.
\newblock \href {http://dx.doi.org/10.1145/1989493.1989505}
  {\path{doi:10.1145/1989493.1989505}}.

\bibitem[LPSP08]{lotker2008improved}
Zvi Lotker, Boaz Patt-Shamir, and Seth Pettie.
\newblock Improved distributed approximate matching.
\newblock In {\em Proceedings of the 20th {ACM} Symposium on Parallelism in
  Algorithms and Architectures ({SPAA})}, pages 129--136, 2008.
\newblock \href {http://dx.doi.org/10.1145/1378533.1378558}
  {\path{doi:10.1145/1378533.1378558}}.

\bibitem[LPSPP05]{Lotker2005MinimumWeightST}
Zvi Lotker, Boaz Patt-Shamir, Elan Pavlov, and David Peleg.
\newblock Minimum-weight spanning tree construction in {$O(\log \log n)$}
  communication rounds.
\newblock {\em {SIAM} Journal on Computing}, 35(1):120--131, 2005.
\newblock \href {http://dx.doi.org/10.1137/s0097539704441848}
  {\path{doi:10.1137/s0097539704441848}}.

\bibitem[Whi12]{White2009HadoopTD}
Tom White.
\newblock {\em Hadoop: The definitive guide}.
\newblock ``O'Reilly Media, Inc.'', 2012.

\bibitem[Wyl79]{Wyllie1979TheCO}
James~C. Wyllie.
\newblock {\em The Complexity of Parallel Computations}.
\newblock PhD thesis, Cornell University, 1979.
\newblock URL: \url{https://hdl.handle.net/1813/7502}.

\bibitem[ZCF{\etalchar{+}}10]{Zaharia2010SparkCC}
Matei Zaharia, Mosharaf Chowdhury, Michael~J. Franklin, Scott Shenker, and Ion
  Stoica.
\newblock Spark: Cluster computing with working sets.
\newblock In {\em Proceedings of the 2nd USENIX Conference on Hot Topics in
  Cloud Computing}, 2010.

\end{thebibliography}

\end{document}